\documentclass{new_tlp}
\usepackage{mathptmx}

\usepackage{amsmath}
\usepackage{amssymb}
\usepackage{mathtools}
\usepackage{stmaryrd} 
\usepackage{xspace}
\usepackage{xcolor} 
\newcommand{\citet}[1]{\cite{#1}} 

\usepackage{listings}
\usepackage{pkg/math-gen}
\usepackage{pkg/is}
\usepackage{pkg/lp} 
\usepackage{pkg/sem}
\usepackage{pkg/languages}

\usepackage{hyperref}
\usepackage[capitalize, noabbrev]{cleveref}

\newif\ifsubmit
\submittrue
\ifsubmit
\newcommand{\EZ}[1]{{#1}} 
\newcommand{\FD}[1]{{#1}} 
\newcommand{\DA}[1]{{#1}} 
\newcommand{\EZComm}[1]{} 
\newcommand{\FDComm}[1]{} 
\newcommand{\DAComm}[1]{} 
\else 
\newcommand{\EZ}[1]{\textcolor{blue}{#1}} 
\newcommand{\FD}[1]{{\color{red}#1}} 
\newcommand{\DA}[1]{\textcolor{magenta}{#1}} 

\newcommand{\EZComm}[1]{{\scriptsize\textcolor{blue}{[\bf{Elena: }#1}]}}
\newcommand{\FDComm}[1]{{\scriptsize\textcolor{red}{[\bf{Francesco: }#1}]}}
\newcommand{\DAComm}[1]{{\scriptsize\textcolor{magenta}{[\bf{Davide: }#1}]}}
\fi

  \title[]{Flexible coinductive logic programming}

  \author[F. Dagnino, D. Ancona, and E. Zucca]
         {FRANCESCO DAGNINO, DAVIDE ANCONA, ELENA ZUCCA\\
         DIBRIS, University of Genova\\
         \email{francesco.dagnino@dibris.unige.it,\{davide.ancona,elena.zucca\}@unige.it}}
         
\newtheorem{theorem}{Theorem}[section] 
\newtheorem{definition}{Definition}[section]
\newtheorem{lemma}{Lemma}[section]
\newtheorem{proposition}{Proposition}[section] 

\begin{document}
%\nocite{*}% includes all entries of BibTeX database into the list of references.

\label{firstpage}

\maketitle
\EZComm{16 pagine inclusa biblio}

  \begin{abstract}
Recursive definitions of predicates are usually interpreted either inductively or coinductively. 
Recently, a more powerful approach has been proposed, called \emph{flexible coinduction}, to express a variety of intermediate interpretations, necessary in some cases to get the correct meaning. 
We {provide a detailed formal account of} an extension of logic programming supporting flexible coinduction. 
Syntactically, programs are enriched by \emph{coclauses}, clauses with a special meaning used to tune the interpretation of predicates. 
As usual, the declarative semantics can be expressed as a fixed point which, however, is not necessarily the least, nor the greatest one, but is determined by the coclauses. 
Correspondingly, the operational semantics is a combination of standard SLD resolution and coSLD resolution. 
We prove that the operational semantics is sound and complete with respect to declarative semantics restricted to finite comodels.  
This paper is under consideration for acceptance in TPLP. 
\end{abstract}


%\tableofcontents

\section{Introduction} 
Standard inductive and coinductive semantics of logic programs sometimes {are} not enough to properly {define} predicates on possibly infinite terms {\cite{SimonBMG07,Ancona13}}. 

Consider the logic program in~\refToFigure{intro-example}, defining some predicates on  lists of numbers represented with the standard Prolog syntax. 
For simplicity, we consider built-in numbers, as in Prolog. 
\begin{figure}
\figrule
\programmath
\[
\begin{array}{lcl}
all\_pos([~]) & \clSep& \\
all\_pos([N|L]) &\clSep& N>0,~all\_pos(L).\\
member(X,[X|\underline{~~}]) &\clSep&\\
member(X,[Y|L]) &\clSep&X\neq Y,~ member(X,L).\\
maxElem([N],N) &\clSep&\\
maxElem([N|L],M) &\clSep& maxElem(L,M_1),~ M \textrm{ is } \max(N,M_1).
\end{array}
\]
\unprogrammath
\caption{An example of logic program: $all\_pos(l)$ succeeds iff $l$ contains only positive numbers,
$member(x,l)$ succeeds iff $x$ is in $l$, $maxElem(l,x)$ succeeds iff $x$ is the greatest number in $l$.}\label{fig:intro-example}
\figrule
\end{figure}

In standard logic programming, terms are inductively defined, that is, are finite, and predicates are inductively defined as well. In the example program, only finite lists are considered, \EZ{such as, e.g., \lstinline{[1|[2|[]]]}, }
and the three predicates are correctly defined on such lists. 
%Formally, the universe is the \emph{Herbrand basis}, and the semantics of a program The inductive interpretation 
%Formally, set $\TOp{\lpProg}$ the inference operator associated to a program $\lpProg$, the semantics of $\lpProg$ is the smallest subset of the universe such that $\TOp{\lpProg}(I) \subseteq I$.   

Coinductive logic programming {(coLP)} \cite{Simon06} extends standard logic programming with the ability of reasoning
about infinite objects and their properties. Terms are coinductively defined, that is, can be infinite, and predicates are coinductively defined as well. In the example, also infinite lists, \EZ{such as \lstinline{[1|[2|[3|[4|...]]]]},}
are considered, and the coinductive interpretation of $all\_pos$ gives the expected meaning on such lists. 
However, this is not the case for the other two predicates: for $member$ the correct interpretation is still the inductive one,  as in the coinductive semantics 
$member(x,l)$ always {succeeds} for an infinite list $l$.  \EZ{For instance, for $L$ the infinite list of $0$'s, $member(1,L)$  has an infinite proof tree where for each node we apply the second clause.}
Therefore, these two predicates cannot coexist in the same program, as they require two different interpretations.\footnote{{To overcome this issue, \emph{co-logic programming} \citet{SimonBMG07} marks predicates as either inductive or coinductive. 
The declarative semantics, however, becomes quite complex, because stratification is needed.}}

The predicate $maxElem$ shows an even worse situation. 
The inductive semantics again does not work on infinite lists, but also the coinductive one is not correct:
$maxElem(l,n)$ {succeeds} whenever $n$ is greater than all the elements of $l$. 
The expected meaning lies between the inductive and the coinductive semantics, hence, to get it, we need something beyond standard semantics. 

Recently,  in the more general context of inference systems \cite{Aczel77}, \emph{flexible coinduction} has been proposed by Dagnino and Ancona et al. \citeyear{Dagnino17,AnconaDZ@esop17,Dagnino19}, a generalisation able to express a variety of intermediate interpretations. As we recall in \cref{sect:inf-sys-lp}, clauses of a logic program can be seen as meta-rules of an inference system where judgments are ground atoms. Inference rules are ground instances of clauses, and a ground atom is valid if it has a finite proof tree in the inductive interpretation, a possibly infinite proof tree in the coinductive one. 

Guided by this abstract view, which provides solid foundations, we develop an extension of logic programming supporting flexible coinduction.   

Syntactically, programs are enriched by \emph{coclauses}, which resemble clauses but have a special meaning used to tune the interpretation of predicates. By adding coclauses, we can obtain a declarative semantics intermediate between the inductive and the coinductive one. Standard (inductive) and {coinductive} logic programming are subsumed by a particular choice of coclauses. 
Correspondingly, operational semantics is a combination of standard SLD resolution \cite{Lloyd87,Apt97} and coSLD resolution as introduced by Simon et al. \citeyear{Simon06,SimonMBG06,SimonBMG07}. 
More precisely, as in coSLD resolution, it keeps trace of already {considered} goals, called \emph{coinductive hypotheses}.\footnote{\EZ{We prefer to mantain this terminology, inherited from coSLD resolution, even though not corresponding to the proof theoretic sense.}}
However, when a goal unifying with a coinductive hypothesis is found, rather than being considered successful as in coSLD resolution, its standard SLD resolution is triggered in the program where also coclauses are considered. Our main result is that such operational semantics is \emph{sound and complete} with respect to the declarative one restricted to regular \mbox{proof trees.}\EZComm{confronto con altre prove di completezza? da discutere cosa dire di Simon; qui non restringiamo la base di Herbrand}

An important additional result is that the operational semantics is not incidental, but, as the declarative semantics, turns out to correspond to a precise notion on the inference system denoted by the logic program.
Indeed, as detailed in a companion paper of Dagnino \citeyear{Dagnino20}, given an inference system, we can always construct another one, with judgments enriched by circular hypotheses, which, interpreted inductively, is equivalent to the \emph{regular} interpretation of the original inference system. In other words, there is a canonical way to derive a (semi-)algorithm to show that a judgment has a regular proof tree, and our operational semantics corresponds to this algorithm.
  This more abstract view supports the reliability of the approach, and, indeed, the proof of equivalence with declarative semantics can be nicely done in a modular way, that is, by relying on a general result proved by Dagnino \citeyear{Dagnino20}.

After basic notions in \cref{sect:inf-sys-lp}, in \cref{sect:coclauses} we introduce logic programs with coclauses and their declarative semantics, and in \cref{sect:lp-opsem} the operational semantics. We provide significant examples in \cref{sect:examples}, the results in \cref{sect:correctness}, related work and conclusive remarks in \cref{sect:related}.

\section{{Logic programs as inference systems}} \label{sect:inf-sys-lp} 

We recall basic concepts about inference systems \cite{Aczel77}, and present (standard inductive and coinductive) logic programming \cite{Lloyd87,Apt97,Simon06,SimonMBG06,SimonBMG07} 
as a particular  instance of this general semantic framework.

\paragraph{Inference systems}
Assume a set $\universe$ called \emph{universe}  whose elements $\judg$ are called \emph{judgements}. 
An \emph{inference system} $\is$ is a set of \emph{(inference) rules}, which are pairs $\RulePair{\prem}{\conclu}$, also written $\Rule{\prem}{\conclu}$, with $\prem \subseteq \universe$ set of \emph{premises}, and $\conclu \in \universe$ \emph{conclusion}. 
We assume inference systems to be \emph{finitary}, that is, rules have a finite set of premises. 
A \emph{proof tree} (a.k.a. \emph{derivation}) in $\is$  is a tree with nodes (labelled) in $\universe$ such that, for each $\judg$ with set of children $\prem$, there is a rule $\RulePair{\prem}{\judg}$ in $\is$. 
{A proof tree for $\judg$ is a proof tree with root $\judg$.}
The \emph{inference operator} $\fun{\InfOp{\is}}{\wp(\universe)}{\wp(\universe)}$ is defined by: 
$$\InfOp{\is}(X) = \{ \judg \in \universe \mid \RulePair{\prem}{\judg} \in \is\ {\mbox{for some}\ \prem \subseteq X}\} $$
A set $X\subseteq \universe$ is \emph{closed} if $\InfOp{\is}(X) \subseteq X$, \emph{consistent} if $X\subseteq \InfOp{\is}(X)$, a \emph{fixed point} if $X = \InfOp{\is}(X)$. 

An \emph{interpretation} of an inference system $\is$ is a set of judgements, that is, a subset of the universe $\universe$. 
%There are two main approaches to define interpretations of an inference system: the model-theoretic and the proof-theoretic one. 
The two standard interpretations, the inductive and the coinductive one, can be defined in {either model-theoretic or} proof-theoretic terms \cite{LeroyG09}.
\begin{itemize}
\item 
The \emph{inductive interpretation} $\Ind{\is}$ is {the intersection of all closed sets, that is, the least closed set}
or,  equivalently, the set of judgements with a finite proof tree.  
\item 
The \emph{coinductive interpretation} $\CoInd{\is}$ is {the union of all consistent sets, that is, the greatest consistent set,}
or,  equivalently,  the set of judgements with an arbitrary (finite or not) proof tree. 
\end{itemize}
By the fixed point theorem \cite{Tarski55}, both $\Ind{\is}$ and $\CoInd{\is}$ are fixed points of $\InfOp{\is}$, the least and the greatest one{, respectively}. 
We will write $\validInd{\is}{\judg}$ for $\judg \in \Ind{\is}$ and $\validCo{\is}{\judg}$ for $\judg \in \CoInd{\is}$. 

\paragraph{Logic programming}
Assume a \emph{first order signature} $\Triple{\PSet}{\FSet}{\VarSet}$ with $\PSet$ set of \emph{predicate symbols} $\psym$, $\FSet$ set of \emph{function symbols} $\fsym$, and $\VarSet$ countably infinite set of \emph{variable symbols} $\lpX$  (\emph{variables} for short). 
Each symbol comes with its \emph{arity}, a natural number denoting the number of arguments. 
Variables have arity $0$. 
A function symbol with arity $0$ is a \emph{constant}. 

\emph{Terms} $t$, $s$, $r$ are (possibly infinite) trees with nodes labeled by function or variable symbols, where the number of children of a node is the symbol arity\footnote{For a more formal definition based on paths see, e.g., \EZ{the work of Ancona and Dovier \citeyear{AnconaD15}}.}. \emph{Atoms} $\Atm$, $\bAtm$, $\cAtm$ are (possibly infinite) trees with the root labeled by a predicate symbol and other nodes by function or variable symbols, again accordingly with the arity. 
Terms and atoms are \emph{ground} if they do not contain variables, and  \emph{finite} (or \emph{syntactic}) if they are finite trees. 
\emph{(Definite) clauses} have shape $\clause{\Atm}{\bAtm_1,\ldots,\bAtm_n}$ with $n\ge 0$, $\Atm$, $\bAtm_1$, \ldots, $\bAtm_n$ finite atoms. 
A clause where $n=0$ is called a \emph{fact}.
A \emph{(definite) logic program} $\lpProg$ is a finite set of clauses. 

\emph{Substitutions} $\sbt,\asbt$ are partial maps from variables to terms with a finite domain.
We write $\appSubst{t}{\sbt}$ for the application of $\sbt$ to a term $t$, call $\appSubst{t}{\sbt}$ an \emph{instance} of $t$, and analogously for atoms{, set of atoms,} and clauses. 
A substitution $\sbt$ is \emph{ground} if, for all $\lpX\in\dom(\sbt)$, $\sbt(\lpX)$ is ground, \emph{syntactic} if, for all $\lpX\in\dom(\sbt)$, $\sbt(\lpX)$ is  a finite (syntactic) term. 

{In order to see a logic program $\lpProg$ as an inference system, we fix as universe the \emph{complete Herbrand base} $\coHB$, that is, the set of all (finite and infinite) ground atoms\footnote{Traditionally \cite{Lloyd87}, the inductive declarative semantics is restricted to finite atoms. 
We define also the inductive semantics on the complete Herbrand base in order to work in a uniform context.}.  Then, $\lpProg$ can be seen as a set of \emph{meta-rules} defining an inference system $\Ground{\lpProg}$ on $\coHB$. That is, $\Ground{\lpProg}$ is the set of ground instances of clauses in $\lpProg$, where $\clause{\Atm}{\bAtm_1,\ldots,\bAtm_n}$ is seen as an inference rule $\RulePair{\{\bAtm_1,\ldots,\bAtm_n\}}{\Atm}$.} 
{In this way, typical notions related to declarative semantics of logic programs turn out to be instances of analogous notions for inference systems. Notably, the
%The \emph{declarative semantics} of a logic program describes its meaning in an abstract way, as the set of ground atoms which are defined to be true by the program, in a sense to be made precise depending on the kind of declarative semantics we choose.
%
%To define declarative semantics, we first need to define the set of  atoms we are interested in. 
%The \emph{complete Herbrand Universe} $\coHU$ \cite{Lloyd87}
%is the set of all (finite and infinite) ground terms.  
%The \emph{complete Herbrand base} $\coHB$ is the set of all (finite and infinite) ground atoms.
(one step) inference operator associated to a program \mbox{$\fun{\TOp{\lpProg}}{\wp(\coHB)}{\wp(\coHB)}$}, defined by:
\[
\TOp{\lpProg}(I)  = \{\Atm\in\coHB \mid (\clause{\Atm}{\bAtm_1,\dots,\bAtm_n}) \in \Ground{\lpProg},  \{\bAtm_1,\dots,\bAtm_n\} \subseteq I\} 
\] is exactly $\InfOp{\Ground{\lpProg}}$. }
%where $\Ground{\lpProg}$ is the set of all ground instances of clauses in $\lpProg$.
%This analogy is much stronger: $\Ground{\lpProg}$ is basically an inference system on $\coHB$ and $\TOp{\lpProg}$ coincides with $\InfOp{\Ground{\lpProg}}$. 
%Indeed, most definitions concerning declarative semantics of logic programs, we report below,  are a rephrasing of analogous concepts for inference systems. 
An \emph{interpretation} (a set $I\subseteq\coHB$) is a \emph{model} of a program $\lpProg$ if $\TOp{\lpProg}(I) \subseteq I$, that is, it is closed with respect to $\Ground{\lpProg}$. 
Dually, an interpretation $I$  is a \emph{comodel} of a program $\lpProg$ if $I\subseteq\TOp{\lpProg}(I)$, that is, it is consistent with respect to $\Ground{\lpProg}$. 
Then, the inductive declarative semantics
 of $\lpProg$ is the least model of $\lpProg$ and the coinductive declarative semantics\footnote{\EZ{Introduced \citet{Simon06,SimonMBG06}} to properly deal with predicates on infinite terms.} is the greatest comodel of $\lpProg$. 
These two semantics coincide with the inductive and coinductive interpretations of $\Ground{\lpProg}$, hence
%, abusing a bit the notation, 
we denote them by $\Ind{\lpProg}$ and $\CoInd{\lpProg}$, respectively.

\section{Coclauses} \label{sect:coclauses} 
{We introduce logic programs with coclauses and define their declarative semantics. Consider again the example in~\refToFigure{intro-example}
%% \begin{lstlisting}
%% all_pos([]) $\clSep$ 
%% all_pos([N|L]) $\clSep$ N>0, all_pos(L).
%% member(X,[X|_]) $\clSep$
%% member(X,[Y|L]) $\clSep$X\=Y, member(X,L).
%% maxElem([N],N) $\clSep$
%% maxElem([N|L],M) $\clSep$ maxElem(L,M1), M is max(N,M1).
%% \end{lstlisting}
{where, as discussed in the Introduction, each predicate needed a different kind of interpretation.}

As shown in the previous section, the above logic program can be seen as an inference system.
In this context, \EZ{\emph{flexible coinduction} has been proposed \citet{Dagnino17,AnconaDZ@esop17,Dagnino19},} a generalisation able to overcome these limitations. 
The key notion  are \emph{corules}, special inference rules used to control the  semantics of an inference system. 
More precisely, a \emph{generalized inference system}, or \emph{inference system with corules}, is a pair of inference systems $\Pair{\is}{\cois}$, where the elements of $\cois$ are called corules. 
The {interpretation} of $\Pair{\is}{\cois}$, denoted by $\FlexCo{\is}{\cois}$, is constructed in two steps. 
\begin{itemize}
\item first, we take the inductive   interpretation of the union $\is\cup\cois$, that is, $\Ind{\is\cup\cois}$, 
\item then, 
%the coinductive interpretation of $\is$ restricted to rules with conclusion in $\Ind{\is\cup\cois}$ (equivalently,
the union of all sets, consistent with respect to $\is$, \EZ{which are subsets of} $\Ind{\is\cup\cois}$, that is,  the largest consistent \EZ{subset of} $\Ind{\is\cup\cois}$.
\end{itemize}
In proof-theoretic terms, $\FlexCo{\is}{\cois}$ is the set of judgements with an arbitrary (finite or not) proof tree in $\is$, whose nodes all have a finite proof tree in $\is {\cup} \cois$. 
Essentially, by corules we filter out some, undesired, infinite proof trees. 
\EZ{Dagnino \citeyear{Dagnino19}} proved that $\FlexCo{\is}{\cois}$ is a fixed point \mbox{of $\InfOp{\is}$}. 

To introduce flexible coinduction in logic programming, first we slightly extend the syntax {by introducing \emph{(definite) coclauses}}, written $\coclause{\Atm}{\bAtm_1,\ldots,\bAtm_n}$, 
where $\Atm$, $\bAtm_1$, \ldots, $\bAtm_n$ are finite atoms. A coclause where n = 0 is called a \emph{cofact}.
Coclauses  syntactically resemble clauses,  but are used in a special way, like corules for inference systems. 
More precisely, we have the following definition:
\begin{definition} \label{def:coclauses}
A \emph{logic program with coclauses} is a pair $\Pair{\lpProg}{\colpProg}$ where $\lpProg$ and $\colpProg$ are sets of clauses.  
{Its \emph{declarative semantics}}, denoted by $\FlexCo{\lpProg}{\colpProg}$,  is the largest comodel of $\lpProg$  \EZ{which is a subset of} $\Ind{\lpProg\cup\colpProg}$. 
\end{definition}
In other words, the declarative semantics of $\Pair{\lpProg}{\colpProg}$ is the coinductive semantics of $\lpProg$ where, however, clauses are instantiated only on elements of $\Ind{\lpProg\cup\colpProg}$. {Note that this is the interpretation of the generalized inference system $\Pair{\Ground{\lpProg}}{\Ground{\colpProg}}$. 

Below is the version of the example in~\refToFigure{intro-example}, equipped with coclauses.
\figrule
\programmath
\[
\begin{array}{lcl}
all\_pos([~]) & \clSep& \\
all\_pos([N|L]) &\clSep& N>0,~all\_pos(L).\\
all\_pos(\underline{~~}) &\coclSep& \\
member(X,[X|\underline{~~}]) &\clSep&\\
member(X,[Y|L]) &\clSep&X\neq Y,~ member(X,L).\\
maxElem([N],N) &\clSep&\\
maxElem([N|L],M) &\clSep& maxElem(L,M_1),~ M \textrm{ is } \max(N,M_1).\\
maxElem([N|\underline{~~}],N) &\coclSep&
\end{array}
\]
\unprogrammath
\figrule
In this way, all the predicate {definitions are correct w.r.t. the expected semantics}: 
\begin{itemize}
\item $all\_pos$ has coinductive semantics, as the coclause allows any infinite proof trees.
\item $member$  has inductive semantics, as without coclauses no infinite proof tree is allowed. 
\item $maxElem$ has an intermediate semantics, as the coclause allows only infinite proof trees where nodes have shape $maxElem(l,x)$ with $x$ an element of $l$. 
\end{itemize}
As the example shows, coclauses allow the programmer to mix inductive and coinductive predicates, and to {correctly define predicates which are neither inductive, nor purely coinductive}. For this reason we call this paradigm \emph{flexible coinductive logic programming.}
%{Note that $\FlexCo{\lpProg}{\colpProg}$ is different from $\CoInd{\lpProg}\cap\Ind{\Extended{\lpProg}{\colpProg}}$. For instance, let $\lpProg$ be the program}
%\begin{lstlisting}
%p(0) $\clSep$ p(0), p(1)
%p(1) $\clSep$ p(0), p(1)
%\end{lstlisting}
%{and $\colpProg$ be the singleton set consisting of the co-fact}
%\begin{lstlisting}
%.(p(0))
%\end{lstlisting}
%{Then, $\CoInd{\lpProg}=\{\texttt{p(0)}, \texttt{p(1)}\}$, and $\Ind{\Extended{\lpProg}{\colpProg}}=\{\texttt{p(0)}\}$. Hence the intersection is $\{\texttt{p(0)}\}$, whereas $\FlexCo{\lpProg}{\colpProg}=\emptyset$.}
%
%From results proved by \citet{AnconaDZ@esop17,Dagnino19} in the more general framework of inference systems, we get that  $\FlexCo{\lpProg}{\colpProg}$ is a fixed point of the operator $\TOp{\lpProg}$ which is neither the greatest, nor the least one, hence it is both a model and a comodel of $\lpProg$. 
Note that, \EZ{as shown  for inference systems with corules \cite{Dagnino17,AnconaDZ@esop17,Dagnino19}, }
inductive and coinductive semantics are particular cases. Indeed, they can be recovered by special choices of coclauses:
\FD{the former is obtained when no coclause is specified, the latter when each atom in $\coHB$ is an instance of the head of a cofact. }

\section{Big-step operational semantics}\label{sect:lp-opsem}

In this section we define an operational counterpart of the {declarative} semantics {of logic programs with coclauses} introduced in the previous section. 

As in standard {coLP} \cite{Simon06,SimonMBG06,SimonBMG07}, to represent possibly infinite terms we use finite sets of equations between finite (syntactic) terms. {For instance, the equation $\lpeq{\texttt{L}}{\texttt{[1,2|L]}}$ represents the infinite list \texttt{[1,2,1,2,\ldots]}.}
%, relying on \emph{unification}. %citazione 
%More precisely, in this way we can only model \emph{regular} terms and atoms, that is, 
%terms having a finite number of different subterms, see, e.g., \cite{Courcelle83,AdamekMV06} for the details.
%This restriction to regular terms and atoms is adopted also in standard coinductive logic programming . %menzionare altri approcci?

Since the declarative semantics  {of logic programs with coclauses} is a combination of inductive and coinductive semantics, 
their operational semantics combines standard SLD resolution \cite{Lloyd87,Apt97} and coSLD resolution \cite{Simon06,SimonMBG06,SimonBMG07}. 
It is presented, rather than in the traditional small-step style, in big-step style, as introduced \EZ{by Ancona and Dovier \citeyear{AnconaD15}}.
This style turns out to be simpler since coinductive hypotheses  (see below) can be kept local. 
Moreover, it naturally leads to an interpreter, and makes it simpler to prove its correctness with respect to declarative semantics (see the next section). 
%For a proof of equivalence of big-step and small-step coSLD resolution see \cite{AnconaD15}.

We introduce some notations. {First of all, in this section we assume atoms and terms to be finite (syntactic).}
{A \emph{goal} is a pair $\extG{\Goal}{\EqSet}$, where $\Goal$ is a finite sequence of atoms.  A goal is \emph{empty} if $\Goal$ is the empty sequence, denoted $\EList$. }
An \emph{equation} has shape $\lpeq{s}{t}$  where $s$ and $t$ are terms, and we denote by $\EqSet$ a finite set of equations. 
%An \emph{extended goal} is a pair $\extG{\Goal}{\EqSet}$.
% where $\Goal$ is a goal and $\EqSet$ is a finite set of equations. 

Intuitively, {a goal} can be seen as a query to the program and the operational semantics has to compute answers (a.k.a. solutions) to such a query. 
More in detail, the operational semantics, given {a goal} $\extG{\Goal}{\EqSet_1}$, provides another set of equations $\EqSet_2$, which represents answers to the {goal}. 
For instance, given the previous program, for the {goal} $\extG{\texttt{maxElem(L,M)}}{\{\lpeq{\texttt{L}}{\texttt{[1,2|L]}}\}}$, the operational semantics {returns} the set of equations $\{\lpeq{\texttt{L}}{\texttt{[1,2|L]}}, \lpeq{{\texttt{M}}}{\texttt{2}}\}$. 

{The judgment of the operational semantics has shape 
$$\colp{\lpProg}{\colpProg}{\cohyp}{\Goal}{\EqSet_1}{\EqSet_2}$$
meaning that resolution of $\extG{\Goal}{\EqSet_1}$, under the \emph{coinductive hypotheses}  $\cohyp$ \cite{SimonMBG06}, succeeds  in $\Pair{\lpProg}{\colpProg}$, producing a set of equations $\EqSet_2$.
Set $\GVar{t}$ the set of variables in a term, and analogously for atoms, set of atoms, and equations. \EZ{We assume $\GVar{\cohyp}\subseteq\GVar{\EqSet_1}$, modelling the intuition that $\cohyp$ keeps track of already {considered} atoms.}
This condition  holds for the initial judgement, and is preserved by rules in \refToFigure{opsem}, hence it is not restrictive.  
Resolution starts with no coinductive hypotheses, that is, the top-level judgment has shape $\colp{\lpProg}{\colpProg}{\emptyset}{\Goal}{\EqSet_1}{\EqSet_2}$. }

{The operational semantics has two flavours:
\begin{itemize}
\item If there are no corules ($\colpProg=\emptyset$), then the judgment models standard SLD resolution, hence the set of coinductive hypotheses is not significant. 
\item Otherwise, the judgment models \emph{flexible coSLD resolution}, which follows the same schema of coSLD resolution, in the sense that it keeps track in $\cohyp$ of the already {considered} atoms. 
However, when {an atom $\Atm$ {in the current goal} unifies with a coinductive hypothesis}, rather than just considering $\Atm$ successful as  in coSLD resolution, standard SLD resolution of $\Atm$ is triggered in the program $\lpProg\cup\colpProg$,  that is, also coclauses can be used. 
\end{itemize}}

%\begin{itemize}
%
%\item $\colpsem{\cohyp}{\Goal}{\EqSet_1}{\EqSet_2}$, meaning that resolution of the {goal} $\extG{\Goal}{\EqSet_1}$, under the \emph{coinductive hypotheses} $\cohyp$, succeeds  in $\Pair{\lpProg}{\colpProg}$, producing a set of equations $\EqSet_2$.
%\item $\lpsem{\Goal}{\EqSet_1}{\EqSet_2}$, meaning that (standard SLD) resolution of the {goal} $\extG{\Goal}{\EqSet_1}$ succeeds  in $\lpProg\cup\colpProg$, producing a set of equations $\EqSet_2$.
%\end{itemize}
\DAComm{commento sulle scelte che le ipotesi coinduttive sono un insieme, come succede per i programmi}
%The top-level judgement is $\colpsem{{\emptycohyp}}{\Goal}{\EqSet_1}{\EqSet_2}$, since resolution starts with no coinductive hypotheses. 

{The judgement is} inductively defined by the rules in \refToFigure{opsem}, which rely on some auxiliary {(standard)} notions. 
A \emph{solution} of an equation $\lpeq{s}{t}$ is a \emph{unifier} \DAComm{forse menzionerei risultato generale, con citazione?, rappresentazione m.g.u. tramite insieme di equazioni e che l'uso di equazioni semplifica la semantica} of $t$ and $s$, that is, a substitution $\sbt$ such that $\appSubst{s}{\sbt} = \appSubst{t}{\sbt}$.
A solution of a finite set of equations $\EqSet$ is a solution of all the equations in $\EqSet$ and $\EqSet$ is \emph{solvable} if there exists a solution of $\EqSet$. 
Two atoms $\Atm$ and $\bAtm$ are \emph{unifiable} in a set of equations $\EqSet$, {written} $\unifiable{\EqSet}{\Atm}{\bAtm}$, if $\Atm = \psym(s_1,\ldots,s_n)$, $\bAtm = \psym(t_1,\ldots,t_n)$ and $\EqSet \cup \{ \lpeq{s_1}{t_1},\ldots, \lpeq{s_n}{t_n} \}$ is solvable, and we denote by $\UnEq{\Atm}{\bAtm}$ the set $\{ \lpeq{s_1}{t_1},\ldots, \lpeq{s_n}{t_n} \}$. 
\begin{figure*}[t]
\[
\begin{array}{c}
\MetaRule{empty}
{  }
{\colp{\lpProg}{\colpProg}{\cohyp}{\EList}{\EqSet}{\EqSet}}
{}
\BigSpace 
\MetaRule{co-hyp}
{  
  \begin{array}{l}
    \colp{\lpProg\cup\colpProg}{\emptyset}{\emptyset}{\Atm}{\EqSet_1\cup \UnEq{\Atm}{\bAtm}}{\EqSet_2}\\
    \colp{\lpProg}{\colpProg}{\cohyp}{\Goal_1,\Goal_2}{\EqSet_2}{\EqSet_3}
  \end{array}
}
{ \colp{\lpProg}{\colpProg}{\cohyp}{\Goal_1,\Atm,\Goal_2}{\EqSet_1}{\EqSet_3} }
{
\bAtm \in \cohyp \\
{\unifiable{\EqSet_1}{\Atm}{\bAtm}} \\
\colpProg\ne\emptyset
}
\\[4ex]
\MetaRule{step}
{  
  \begin{array}{l}
    \colp{\lpProg}{\colpProg}{\cohyp\cup \{\Atm\}}{\cAtm_1,\dots,\cAtm_n}{\EqSet_1\cup \UnEq{\Atm}{\bAtm}}{\EqSet_2} \\
    \colp{\lpProg}{\colpProg}{\cohyp}{\Goal_1,\Goal_2}{\EqSet_2}{\EqSet_3} 
  \end{array}
}
{ \colp{\lpProg}{\colpProg}{\cohyp}{\Goal_1,\Atm,\Goal_2}{\EqSet_1}{\EqSet_3} }
{
\mbox{$\sbt$ fresh renaming} \\
\clause{\appSubst{\bAtm}{\sbt}}{\appSubst{\cAtm_1}{\sbt},\ldots,\appSubst{\cAtm_n}{\sbt}} \in \lpProg \\
{\unifiable{\EqSet_1}{\Atm}{\bAtm}}
}
\end{array}
\]
\caption{Big-step operational semantics}\label{fig:opsem}
\end{figure*}

Rule \rn{empty} states that the  resolution of {an} empty goal succeeds. 
In rule \rn{step}, an atom $\Atm$ {to be resolved is selected}, and a clause of the program is chosen such that $\Atm$ unifies with the head of the clause in the current set of equations.  
Then, resolution of the original goal succeeds if both the body of the selected clause and the remaining atoms are resolved{, enriching the set of equations correspondingly.}
{As customary,} the selected clause is renamed using fresh variables, to avoid variable clashes in the set of equations obtained after unification. 
{Note that, in the resolution of the body of the clause, the selected atom is added to the current set of coinductive hypotheses. This is not relevant for standard SLD resolution ($\colpProg=\emptyset$). However, if $\colpProg\neq\emptyset$, this allows rule \rn{co-hyp} to} handle the case when
{an atom $\Atm$ that has to be resolved unifies with a coinductive hypothesis  in the current} set of equations. 
In this case, standard SLD resolution of such atom in the program $\lpProg\cup\colpProg$ is triggered,  and resolution of the original goal succeeds if both such standard SLD resolution of the selected atom and resolution of the remaining goal succeed. 

\newcommand{\eqL}{\mathit{eq}_\texttt{L}}
\newcommand{\eqs}{\EZ{\mathit{eqs}}}
\newcommand{\lpeqnarrow}[2]{#1{\eqcirc}#2} 

In \refToFigure{example} we show an example of resolution. 
We use the shorter syntax \texttt{=max}, abbreviate by $\eqL$ the equation $\lpeq{\texttt{L}}{\texttt{[1,2|L]}}$,  \EZ{by $\eqs$ the equations $\lpeqnarrow{M3}{2},\lpeqnarrow{M2}{2}$}, by \texttt{mE} the predicate \texttt{maxElem}, and by \refToRule{s}, \refToRule{c} the rules \refToRule{step} and \refToRule{co-hyp}, respectively.  When applying rule \refToRule{step}, we also indicate the clause/coclause which has been used: we write 1,2,3 for the two clauses and the coclause for the \texttt{maxElem} predicate (the first clause is never used in this example). Finally, to keep the example readable and focus on key aspects, we make some simplifications: notably, \refToRule{max} stands for an omitted proof tree solving atoms of shape \texttt{\_ is max(\_,\_)}; morever, equations between lists are implicitly applied.

\begin{figure}
\begin{footnotesize}
$
\MetaRule{s-2}
{    
\MetaRule{s-2}{
\MetaRule{c}
{
\MetaRule{s-2}{
\MetaRule{s-3]}{\MetaRule{max}{}{}{}}{\colpNarrow{\{1,2,3\}}{\emptyset}{\emptyset}{\texttt{mE([2|L],M3)},\texttt{M2=max(1,M3)}}{\eqL,\lpeqnarrow{M2}{M}}{\eqL,\lpeqnarrow{M2}{M},\eqs}}{}
}{
\colpNarrow{\{1,2,3\}}{\emptyset}{\emptyset}{\texttt{mE(L,M2)}}{\eqL,\lpeqnarrow{M2}{M}}{\eqL,\lpeqnarrow{M2}{M},\eqs}
}{}
\MetaRule{max}{}{}{}}{
\colpNarrow{\{1,2\}}{3}{\texttt{mE(L,M)}}{\texttt{mE(L,M2)},\texttt{M1=max(2,M2)}}{\eqL}{\eqL,\lpeqnarrow{M2}{M},\eqs,\lpeqnarrow{M1}{2}}
}{}
}
{\colpNarrow{\{1,2\}}{3}{\texttt{mE(L,M)}}{\texttt{mE([2|L],M1)},\texttt{M=max(1,M1)}}{\eqL}{\eqL,\lpeqnarrow{M2}{M},\eqs,\lpeqnarrow{M1}{2}}\Space\MetaRule{max}{}{}{}
}{}
}
{ \colpNarrow{\{1,2\}}{3}{\emptyset}{\texttt{mE(L,M)}}{\eqL}{\eqL,\lpeqnarrow{M2}{M},\eqs,\lpeqnarrow{M1}{2},\lpeqnarrow{M}{2}} }
{}
$
\end{footnotesize}
\caption{Example of resolution}\label{fig:example}
\end{figure}

As final remark, note that flexible coSLD resolution nicely subsumes both SLD and coSLD. The former, as already said, is obtained when the set of coclauses is empty, that is, the program is inductive. The latter is obtained when, for all predicate $\psym$ of arity $n$, we have a cofact  $\psym(\lpX_1,\ldots,\lpX_n)$.

\section{Examples}\label{sect:examples}
In this section we discuss some more sophisticated examples. 

\paragraph{$\infty$-regular expressions:}
We define \emph{$\infty$-regular expressions} on an alphabet $\Sigma$, a variant of the formalism defined \EZ{by L\"{o}ding and Tollk\"{o}tter} \citeyear{LodingT16}} for denoting languages of finite and infinite words{, the latter also called $\omega$-words,} as follows: 
\[ r ::= \emptyset \mid  \epsilon \mid a \mid r_1\cdot r_2 \mid r_1 + r_2 \mid r^\star \mid r^\omega \]
where $a\in\Sigma$. 
The syntax of standard regular expressions is extended by $r^\omega$, denoting the $\omega$-power of the language $A_r$ denoted by $r$. That is, the {set of words obtained by concatenating infinitely many times words in $A_r$.}
In this way, we can denote also languages containing infinite words. 

In~\refToFigure{regex-example} we define the predicate $match$, such that $match(W,R)$ holds if the finite or infinite word $W$, implemented as a list, belongs to the language denoted by $R$. 
For simplicity, we consider words over the alphabet $\{0,1\}$. 
%The definition is the following: 

\begin{figure}
\figrule
\programmath
\[
\begin{array}{lcl}
concat([~],W,W) &\clSep& \\ 
concat([B|W_1],W_2,[B|W_3]) &\clSep& concat(W_1,W_2,W_3). \\  
concat(W_1,W_2,W_1) &\coclSep& \\ 
match([~],eps) &\clSep&\\
match([0],0) &\clSep&\\
match([1],1) &\clSep& \\
match(W,cat(R_1,R_2)) &\clSep& match(W_1,R_1),~match(W_2,R_2),~concat(W_1,W_2,W).\\
match(W,plus(R_1,R_2)) &\clSep& match(W,R_1).\\
match(W,plus(R_1,R_2)) &\clSep& match(W,R_2).\\
match(W,star(R)) &\clSep& match\_star(N,W,R). \\
match([~],omega(R)) &\clSep& match([~],R). \\
match([B|W],omega(R)) 
  &\clSep& match([B|W_1],R),~match(W_2,omega(R)),~concat(W_1,W_2,W).\\
match(W,omega(R)) &\coclSep& \\
match\_star(0,[~],R) &\clSep& \\
match\_star(s(N),W,R) 
  &\clSep& match(W_1,R),~match\_star(N,W_2,R),~concat(W_1,W_2,W).
\end{array}
\]
\unprogrammath
\caption{A logic program for $\infty$-regular expression recognition.}\label{fig:regex-example}
\figrule
\end{figure}
Concatenation of words needs to be defined coinductively, to correctly work on infinite words as well. 
Note that, when $w_1$ is infinite, $w_1w_2$ is equal to $w_1$.

On operators of regular expressions, $match$ can be defined in the standard way (no coclauses). 
{In particular, the} definition for expressions of shape $r^\star$ follows the explicit  definition of the $\star$-closure of a language:
given a language $L$, a word $w$ belongs to $L^\star$ iff it can be decomposed as $w_1\ldots w_n$, for some $n\ge 0$, where $n=0$ means $w$ is empty, and $w_i\in L$, for all $i\in 1..n$. 
This condition is checked by the auxiliary predicate $match\_star$. 

To define when a word $w$ matches $r^\omega$ we have two cases.
If $w$ is empty, then it is enough to check that the empty word matches $r$, as expressed by the first clause, because concatenating infinitely many times the empty word we get again the empty word.
Otherwise, we have to decompose $w$ as $w_1w_2$ where $w_1$ is not empty and matches $r$ and $w_2$ matches $r^\omega$ as well, 
as formally expressed by the second clause,  
To propertly handle infinite words, we need to concatenate infinitely many non-empty words, hence we need to apply the second clause infinitely many times. 
The coclause allows all such infinite derivations. 

\paragraph{An LTL fragment:}
In~\refToFigure{ltl-example} we define the predicate $sat$ s.t. $sat(w,\varphi)$  
succeeds iff the $\omega$-word $w$ over the alphabet $\{0,1\}$ satisfies the formula $\varphi$ of the fragment of the Linear Temporal Logic
with the temporal operators $until$ ($\mathbf{U}$) and $always$ ($\mathbf{G}$) and the predicate $zero$
and its negation\footnote{Predicates $true$ and $false$ could be easily defined as well.} $one$.
\begin{figure}
\figrule
\programmath
\[
\begin{array}{lcl}
sat\_exists(0,W,Ph) &\clSep& sat(W,Ph). \\
sat\_exists(s(N),[B|W],Ph) &\clSep& sat\_exists(N,W,Ph). \\ 
sat\_all(0,W,Ph) &\clSep& \\
sat\_all(s(N),[B|W],Ph) &\clSep& sat([B|W],Ph),~sat\_all(N,W,Ph). \\ 
sat([0|W],zero) &\clSep&\\
sat([1|W],one) &\clSep&\\
sat([B|W],always(Ph)) &\clSep& sat([B|W],Ph),~sat(W,always(Ph)). \\ 
sat(W,always(Ph)) &\coclSep& \\
sat([B|W],until(Ph_1,Ph_2)) 
  &\clSep& sat\_exists(N,[B|W],Ph_2),~sat\_all(N,[B|W],Ph_1). 
\end{array}
\]
\unprogrammath
\caption{A logic program for satisfaction of an LTL fragment:
  $sat\_exists(N,W,Ph)$ succeeds iff suffix at $N$ of $\omega$-word $W$ satisfies $Ph$,
  $sat\_all(N,W,Ph)$ succeeds iff all suffixes of word $W$ at index $< N$ satisfy $Ph$,
  $sat(W,Ph)$ succeeds iff $\omega$-word $W$ satisfies $Ph$.
}\label{fig:ltl-example}
\figrule
\end{figure}
%% Other option:
%% \begin{lstlisting}[basicstyle=\ttfamily\small]
%% %%% suffix(N,W,S): S is the suffix of word W at N 
%% suffix(0,W,W) $\clSep$
%% suffix(s(N),[B|W],S) $\clSep$ suffix(N,W,S)  
%% %%% sat_all(N,W,Ph): all suffixes of word W at index < N satisfy Ph
%% sat_all(0,W,Ph) $\clSep$
%% sat_all(s(N),[B|W],Ph) $\clSep$ sat([B|W],Ph), sat_all(N,W,Ph) 
%% %%% sat(W,Ph): omega-word W satisfies Ph
%% sat([0|W],zero) $\clSep$
%% sat([1|W],one) $\clSep$
%% sat([B|W],always(Ph)) $\clSep$ sat([B|W],Ph), sat(W,always(Ph)) 
%% sat(W,always(Ph)) $\coclSep$ 
%% sat([B|W],until(Ph1,Ph2)) $\clSep$
%%     suffix(N,[B|W],S), sat_all(N,[B|W],Ph1), sat(S,Ph2) 
%% \end{lstlisting}
%% more compact def. without concat
%% sat_all([],W,Phi1,Phi2) $\clSep$ sat(W,Phi2)
%% sat_all([B|P],[B|W],Phi1,Phi2) $\clSep$ sat([B|W],Phi1), until(P,W,Phi1,Phi2)
%% sat(W,until(Phi1,Phi2)) :- sat_all(P,W,Phi1,Phi2)
Since $sat([B|W],always(Ph))$ succeeds iff all \textbf{infinite} suffixes of $[B|W]$ satisfy formula $Ph$,
the coinductive interpretation has to be considered, hence a coclause is needed; for instance, $sat(W_0,always(zero))$, with $W_0=[0|W_0]$, succeeds because
the atom $sat(W_0,always(zero))$ in the body of the clause for $always$ unifies\footnote{Actually, in this case
  the atom to be resolved and the coinductive hypothesis are syntactically equal.} with
the coinductive hypothesis $sat(W_0,always(zero))$ (see rule \rn{co-hyp} in \cref{fig:opsem})
and the coclause allows it to succeed w.r.t. standard SLD resolution (indeed, atom $sat(W_0,zero)$ succeeds, thanks to
the first fact in the logic program). 

Differently to $always$, the interpretation of $until$ has to be inductive because
$until(\varphi_1,\varphi_2)$ succeeds iff $\varphi_2$ is satisfied after a \textbf{finite} number of steps; for this reason,
no coclause is given for this operator; for instance,
$sat([1,1,0|W_1],until(one,zero))$ with $W_1=[1|W_1]$ succeeds w.r.t. standard SLD resolution, while
$sat(W_1,until(one,zero))$, $sat(W_1,until(always(one),zero))$,
and $sat(W_1,until(always(one),always(zero)))$ fail.
The clause for $sat([B|W],until(Ph_1,Ph_2))$ follows the standard definition of satisfaction for the $\mathbf{U}$ operator:
there must exist a suffix of $[B|W]$ at index $N$ satisfying $Ph_2$ ($sat\_exists(N,[B|W],Ph_2)$) s.t.
all suffixes of $[B|W]$ at index less than $N$ satisfy $Ph_1$ ($sat\_all(N,[B|W],Ph_1)$).
%The only difference with the predicate \lstinline{concat} defined in the previous example on $\infty$-regular expressions is that here we need to consider finite prefixes only, hence no corule is given.

An interesting example concerns the goal $sat([1,1|W_0],until(one,always(zero)))$, where
the two temporal operators are mixed together: it succeeds as expected, thanks to the two clauses
for $until$ and the fact that  $sat(W_0,always(zero))$ succeeds, as shown above.

\EZ{Some of the issues faced in this example are also discussed by Gupta et al.\ \citeyear{GuptaSDMMK11}}.

%% LTL satisfaction is an example that shows that in some cases corules which are not simply cofacts \EZComm{gi\`a mostrato in esempio precedente se resta cos\`i, cambiare un po' frase}
%% are needed to correctly define a predicate; for instance, if the corule of the logic program above is replaced with
%% the cofact \lstinline{sat(W,always(Ph))$\coclSep$}, then 
%% the goal \lstinline{sat(W1,until(one,always(zero)))} with \lstinline{W1=[1|W1]} unsoundly succeeds.

\paragraph{Big-step semantics modeling infinite behaviour and observations} 
Defining a big-step operational semantics modelling divergence is a difficult task, especially in presence of observations. 
Ancona et al. \citeNN{AnconaDZ@ecoop18,AnconaDRZ20}
show how corules can be successfully employed to tackle this problem, by \EZ{providing} big-step semantics able to model divergence
for several variations of the lambda-calculus and different kinds of observations.
Following this approach, we present in~\refToFigure{big-step-example} a similar example, but simpler, to keep it shorter: a logic program with coclauses defining the big-step semantics of a toy language to output possibly infinite sequences\footnote{For simplicity we consider only integers, but in fact the definition below allows any term as output.} of integers. 
%We refer to these works for more sophisticated examples. 
Expressions are regular terms generated by the following grammar: 
\[ \e ::= skip \mid \Write{n} \mid seq(\e_1,\e_2) \]  
where $skip$ is the idle expression, 
$\Write{n}$ outputs $n$, and 
$seq(\e_1,\e_2)$ is the sequential composition. 
The semantic judgement has shape $\evalobs{\e}{r}{s}$, represented by the atom $eval(\e,r,s)$, where 
$\e$ is an expression,
$r$ is either $end$ or $div$, for converging or diverging computations, respectively, and 
$s$ is a possibly infinite sequence of integers. 
\begin{figure}
\figrule
\programmath
\[
\begin{array}{lcl}
concat([~],S,S) &\clSep& \\
concat([N|S_1],S_2,[N|S_3]) &\clSep& concat(S_1,S_2,S_3). \\
eval(skip,end,[~]) &\clSep& \\
eval(out(N),end,[N]) &\clSep& \\
eval(seq(E_1,E_2),R,S) &\clSep& eval(E_1,end,S_1),~eval(E_2,R,S_2),~concat(S_1,S_2,S).\\
eval(seq(E_1,E_2),div,S) &\clSep& eval(E_1,div,S). \\
eval(E,div,[~]) &\coclSep& \\
eval(seq(E_1,E_2),div,S) &\coclSep& eval(E_1,end,[N|S_1]),~concat([N|S_1],S_2,S).
\end{array}
\]
\unprogrammath
\caption{A logic program defining a big-step semantics with infinite behaviour and observations.}\label{fig:big-step-example}
\figrule
\end{figure}
Clauses for $concat$ are pretty standard; in this case the definition is purely inductive (hence, no coclause is needed) 
since the left operand of concatenation is always a finite sequence. 
Clauses for $eval$ are rather straightforward, but sequential composition $seq(\e_1,\e_2)$ deserves some comment:
if the evaluation of $\e_1$ converges, then the computation can continue with the evaluation of $\e_2$, otherwise
the overall computation diverges and $\e_2$ is not evaluated. 

As opposite to the previous examples, here we do not need just cofacts, but also a coclause;
both the cofact and the coclause ensure that for infinite derivations only $div$  can be derived.
Furthermore, the cofact handles diverging expressions which produce a finite output sequence, as in
$eval(E,div,[~])$ or in $eval(seq(out(1),E),div,[1])$, with $E=seq(skip,E)$ or $E=seq(E,E)$, while
the coclause deals with diverging expressions with infinite outputs, as in $eval(E,div,S)$ with $E=seq(out(1),E)$ and
$S=[1|S]$. The body of the coclause ensures that the left operand of sequential composition converges, thus ensuring
a correct productive definition.

\section{{Soundness and completeness}} \label{sect:correctness} 

{After formally relating the two approaches, we state soundness of the operational semantics with respect to the declarative one. 
Then, we show that completeness does not hold in general, and define the \emph{regular} version of the declarative semantics. 
Finally, we show that the operational semantics is equivalent to this restricted declarative semantics.  }

\paragraph{Relation between operational and declarative semantics} {As in the standard case, the first step is to bridge the gap between the two approaches: the former computing equations, the latter defining truth of atoms. This can be achieved through the notions of \emph{answers} to a goal.  }

%Assume a logic program with coclauses $\Pair{\lpProg}{\colpProg}$. 
Given a set of equations $\EqSet$, $\gsol{\EqSet}$ is the set of the \emph{solutions} of $\EqSet$, that is, the ground substitutions unifying all the equations in $\EqSet$. 
{Then, $\sbt\in\gsol{\EqSet}$ is an \emph{answer} to $\extG{\Goal}{\EqSet}$ if \mbox{$\GVar{\Goal}\subseteq\dom(\sbt)$}. }

{The judgment $\colp{\lpProg}{\colpProg}{\cohyp}{\Goal}{\EqSet_1}{\EqSet_2}$ described in \cref{sect:lp-opsem} computes a set of answers to the input goal. Indeed, solutions of the output set of equations are solutions of the input set as well, since the following proposition holds}.

%This assertion is correct thanks to the next proposition, ensuring that $\gsol{\EqSet_2}\subseteq\gsol{\EqSet_1}$: 

\begin{proposition}\label{prop:opsem-monotone} 
\begin{enumerate}
\item If $\colp{\lpProg}{\colpProg}{\cohyp}{\Goal}{\EqSet_1}{\EqSet_2}$ then $\EqSet_1\subseteq\EqSet_2$ and $\GVar{\Goal}\subseteq\GVar{\EqSet_2}$.
\item If $\EqSet_1\subseteq\EqSet_2$, then $\gsol{\EqSet_2}\subseteq \gsol{\EqSet_1}$. 
\end{enumerate}
\end{proposition}
\begin{proof}
(1) Straightforward induction on rules in \cref{fig:opsem}. (2) Trivial.
\end{proof}

{On the other hand, we can define which answers are correct\EZComm{forse userei ovunque ``valid''} in an interpretation:} 
\begin{definition}\label{def:ans}
%Let $\extG{\Goal}{\EqSet}$ be a goal. 
For $I\subseteq\coHB$, the set of answers to $\extG{\Goal}{\EqSet}$ \emph{correct in $I$}  is $\lpAns{\Goal}{\EqSet}{I} = \{ \sbt \in \gsol{\EqSet} \mid \appSubst{\Goal}{\sbt} \subseteq I \}$. 
\end{definition}
%Note that, by \cref{def:coclauses}, we have $\lpans{\Goal}{\EqSet}\subseteq \lpbdans{\Goal}{\EqSet}$. 

{Hence, soundness of the operational semantics can be expressed as follows: all the answers computed for a given goal are correct in the declarative semantics.} 
\begin{theorem}[Soundness w.r.t.\ declarative semantics]\label{thm:lp-sound}
If $\colp{\lpProg}{\colpProg}{\emptyset}{\Goal}{\EqSet}{\EqSet'}$ holds, then $\gsol{\EqSet'} \subseteq \lpAns{\Goal}{\EqSet}{\FlexCo{\lpProg}{\colpProg}}$. 
\end{theorem}

\paragraph{Completeness issues} The converse of this theorem, that is, all correct answers can be computed, cannot hold in general\EZComm{ossia per $I=\FlexCo{\lpProg}{\colpProg}$}, since, as shown \EZ{by Ancona and Dovier \citeyear{AnconaD15}}, coinductive declarative semantics does not admit any complete procedure\footnote{\EZ{That is, establishing whether an atom belongs to the coinductive declarative semantics is neither decidable nor semi-decidable, even when the Herbrand universe is restricted to the set of rational terms.}}, hence our model as well, since it generalizes the coinductive one. 
To explain why completeness does not hold in our case, we can adapt the following example from \EZ{Ancona and Dovier \citeyear{AnconaD15}}\footnote{\FD{Example 10 at page 8.}},
%%Let us analyse an example.  
where $p$ is a predicate symbol of arity $1$, $z$ and $s$ are function symbols of arity $0$ and $1$ respectively.
\figrule
\programmath
\[
\begin{array}{lcl}
p(X) &\clSep& p(s(X)). \\
p(X) &\coclSep&
\end{array}
\]
\unprogrammath
\figrule
Let us define $\underline{0} = z$, $\underline{n+1} = s(\underline{n})$ and $\underline{\omega} = s(s(\ldots))$. 
The declarative semantics is the set $\{ p(\underline{x}) \mid x \in \N\cup\{\omega\}\}$. 
In the operational semantics, instead, only $p(\underline{\omega})$ is considered true. 
Indeed, all derivations have to apply the rule \rn{co-hyp}, which imposes the equation $\lpeq{X}{s(X)}$, whose unique solution is $\underline{\omega}$. 
Therefore, the operational semantics is not complete. 

Now the question is the following:\EZComm{cercherei di dirlo un po' diversamente, da discutere}
can we characterize in a declarative way answers computed by the big-step semantics? 
{In the example, there is a difference between the atoms $p(\underline{\omega})$ and $p(\underline{n})$, with $n\in\N$, because} 
the former has a regular proof tree, namely, a tree with finitely many different subtrees, while the latter has only with non-regular, thus infinite, proof trees. 

Following this observation,
we prove that the operational semantics is sound and complete with respect to the restriction of the declarative semantics to atoms derivable by regular proof trees. 
As we will see, this set can be defined in model-theoretic terms, by restricting to finite comodels of the program. 
\EZ{Dagnino \citeyear{Dagnino20} defined this restriction} for an arbitrary (generalized) inference system. 
We report here relevant definitions and results. 

\paragraph{Regular declarative semantics}
{Let us write $X\subseteqfin Y$ if $X$ is a finite subset of $Y$.}
The \emph{regular interpretation} of $\Pair{\is}{\cois}$ is defined as 
\[ \FlexReg{\is}{\cois} = \bigcup \{X\subseteqfin \Ind{\is\cup\cois} \mid X\subseteq \InfOp{\is}(X)\} \] 

This definition is like the one of $\FlexCo{\is}{\cois}$, except that we take the union\footnote{Which could be an infinite set, hence it is not the same of the greatest finite consistent set.}
 only of those consistent subsets of $\Ind{\is\cup\cois}$ which are \emph{finite}.The set $\FlexReg{\is}{\cois}$ is a fixed point of $\InfOp{\is}$ and, precisely, it is the \emph{rational fixed point} \cite{AdamekMV06} of $\InfOp{\is}$ restricted to $\wp(\Ind{\is\cup\cois})$, hence we get $\FlexReg{\is}{\cois} \subseteq \FlexCo{\is}{\cois}$. 

The proof-theoretic characterization relies on \emph{regular proof trees}, which are proof trees with a finite number of subtrees \cite{Courcelle83}. 
That is, as proved \EZ{by Dagnino \citeyear{Dagnino20},} $\FlexReg{\is}{\cois}$ is the set of judgments with 
a regular proof tree in $\is$ whose nodes all have a finite \mbox{proof tree in $\is\cup\cois$.}

{As special case, we get} regular semantics of logic programs with coclauses. 
\begin{definition}\label{def:lp-reg}
The \emph{regular declarative semantics} of  $\Pair{\lpProg}{\colpProg}$, denoted by $\FlexReg{\lpProg}{\colpProg}$, is the union of all finite comodels included in $\Ind{\lpProg\cup\colpProg}$. 
%The set of \emph{regularly correct answers} to a goal $\extG{\Goal}{\EqSet}$ is the set $\lpregans{\Goal}{\EqSet} = \{ \sbt\in\gsol{\EqSet} \mid \appSubst{\Goal}{\sbt} \subseteq \FlexReg{\lpProg}{\colpProg}\}$. 
\end{definition}
As above, %$\FlexReg{\lpProg}{\colpProg}$ is a fixed point of $\TOp{\lpProg}$ below \EZComm{non mi pare usato prima} $\Ind{\lpProg\cup\colpProg}$, hence we get 
$\FlexReg{\lpProg}{\colpProg}\subseteq\FlexCo{\lpProg}{\colpProg}$, hence $\lpregans{\Goal}{\EqSet}\subseteq \lpans{\Goal}{\EqSet}$. 

{We state now  soundness and completeness of the operational semantics with respect to this semantics. {We write $\sbt\sbtord\asbt$ iff $\dom(\sbt)\subseteq\dom(\asbt)$ and, for all $\lpX\in\dom(\sbt)$, $\sbt(\lpX)=\asbt(\lpX)$.\EZComm{\`e l'inclusione insiemistica} It is easy to see that $\sbtord$ is a partial order and, if $\sbt\sbtord\asbt$ and $\GVar{\Goal}\subseteq\dom(\sbt)$, then $\appSubst{\Goal}{\sbt}=\appSubst{\Goal}{\asbt}$.}}

{\begin{theorem}[Soundness w.r.t.\ regular declarative semantics]\label{thm:reg-lp-sound}
%Let $\extG{\Goal}{\EqSet}$ be a goal, then:
If $\colp{\lpProg}{\colpProg}{\emptyset}{\Goal}{\EqSet}{\EqSet'}$, and $\sbt\in\gsol{\EqSet'}$, then $\sbt\in\lpregans{\Goal}{\EqSet}$.
\end{theorem}
\begin{theorem}[Completeness  w.r.t.\ regular declarative semantics]\label{thm:reg-lp-complete}
If $\sbt\in\lpregans{\Goal}{\EqSet}$, then
 $\colp{\lpProg}{\colpProg}{\emptyset}{\Goal}{\EqSet}{\EqSet'}$, and  
$\sbt\sbtord\asbt$ for some $\EqSet'$ and \mbox{$\asbt\in\gsol{\EqSet'}$}. 
\end{theorem}
}

{That is, any answer computed for a given goal is correct in the regular declarative semantics, and any correct answer is included in a computed answer.   }
\cref{thm:reg-lp-sound} immediately entails \cref{thm:lp-sound} as $\lpregans{\Goal}{\EqSet}\subseteq \lpans{\Goal}{\EqSet}$. 

%The completeness result states that all regularly correct answers can be computed by the big-step semantics: 
%\begin{theorem}\label{thm:reg-lp-complete}
%Let $\extG{\Goal}{\EqSet}$ be a goal, then, 
%for all $\sbt\in\lpregans{\Goal}{\EqSet}$, 
%there exists $\EqSet'$ such that 
%$\colp{\lpProg}{\colpProg}{\emptyset}{\Goal}{\EqSet}{\EqSet'}$ and 
%$\sbt\sbtord\asbt$, for some $\asbt\in\gsol{\EqSet'}$. 
%\end{theorem}

\paragraph{Proof technique} 
{In order to prove the equivalence of the two semantics, we rely on a property which holds in general for the regular interpretation \cite{Dagnino20}: we can construct an equivalent inductive characterization. }
{That is, given a generalized inference system $\Pair{\is}{\cois}$ on the universe $\universe$, we can construct an inference system $\Loopcois{\is}{\cois}$ with judgments of shape $\LoopJ{\LoopHp}{\judg}$, for $\judg\in\universe$ and \mbox{$\LoopHp\subseteqfin\universe$}, such that the inductive interpretation of $\Loopcois{\is}{\cois}$ coincides with the regular interpretation of $\Pair{\is}{\cois}$.}
The set $\LoopHp$, whose elements are called \emph{coinductive hypotheses} \EZComm{forse usare sempre circular}, is used to 
%track already {considered} judgements so that we can 
detect cycles in the proof. 

{In particular, for logic programs with coclauses, we get an inference system with judgments of shape $\LoopJ{\cohyp}{\Atm}$, for $\cohyp$ finite set of ground atoms, and $\Atm$ ground atom, defined as follows.}

{\begin{definition}\label{def:loop-is}
Given $\Pair{\lpProg}{\colpProg}$, the inference system $\Loopcois{\lpProg}{\colpProg}$ consists of the following (meta-)rules:\EZComm{qui ho usato lo stile a frazione perch\'e aumenta moltissimo la leggibilit\`a; userei per le regole nomi analoghi a quelle della sem. op.}
\begin{description}
\item [\rn{hp}] $\Rule{}{\LoopJ{\cohyp}{\Atm}}$ \Space$\Atm \in \cohyp$ and $\Atm \in \Ind{\lpProg\cup\colpProg}$\\
\item [\rn{rule}] $\Rule{\LoopJ{\cohyp\cup\{\Atm\}}{\bAtm_1}\ \ldots\ \LoopJ{\cohyp\cup\{\Atm\}}{\bAtm_n}}{\LoopJ{\cohyp}{\Atm}}$\Space $(\clause{\Atm}{\bAtm_1,\dots,\bAtm_n}) \in \Ground{\lpProg}$
\end{description}
\end{definition}}
The following proposition states the equivalence with the regular interpretation. The proof is given \EZ{by Dagnino \citeyear{Dagnino20}} in the general case of inference systems with corules.
\begin{proposition}\label{prop:loop-is}
$\validInd{\Loopcois{\lpProg}{\colpProg}}{\LoopJ{\emptyset}{\Atm}}$ iff $\Atm \in \FlexReg{\lpProg}{\colpProg}$.
\end{proposition}

{Note that the definition of $\validInd{\Loopcois{\lpProg}{\colpProg}}{\LoopJ{\cohyp}{\Atm}}$ has many analogies with that of the operational semantics in \cref{fig:opsem}. The key difference is that the former handles \emph{ground}, not necessarily finite, atoms, the latter not necessarily ground finite atoms (we use the same metavariables $\Atm$ and $\cohyp$ for simplicity).  In both cases already {considered} atoms are kept in an auxiliary set $\cohyp$. In the former, to derive an atom $\Atm\in\cohyp$,  the side condition requires $\Atm$ to belong to the inductive intepretation of the program $\lpProg\cup\colpProg$. In the latter, when an atom $\Atm$ \emph{unifies}  with one in $\cohyp$, standard SLD resolution is triggered in the program  $\lpProg\cup\colpProg$.  } 

{To summarize, $\validInd{\Loopcois{\lpProg}{\colpProg}}{\LoopJ{\cohyp}{\Atm}}$ can be seen as an abstract version, at the level of the underlying inference system, of operational semantics. 
Hence, the proof of soundness and completeness can be based on proving a precise correspondence between these two inference systems, both interpreted inductively. This is very convenient since the proof can be driven in both directions by induction on the defining rules.}   

{The correspondence is formally stated in the following two lemmas.}

\begin{lemma}[Soundness w.r.t.\ inductive characterization of regular semantics]\label{lem:reg-lp-sound}
For all $\cohyp$ and $\extG{\Atm_1,\ldots,\Atm_n}{\EqSet}$,\\
%Let $\extG{\Goal}{\EqSet}$ be a goal with $\Goal = $ ($n\ge 0$). 
if $\colp{\lpProg}{\colpProg}{\cohyp}{\Atm_1,\ldots,\Atm_n}{\EqSet}{\EqSet'}$ then, for all $\sbt \in \gsol{\EqSet'}$ and $i\in1..n$, $\validInd{\Loopcois{\lpProg}{\colpProg}}{\LoopJ{\appSubst{\cohyp}{\sbt}}{\appSubst{\Atm_i}{\sbt}}}$. 
\end{lemma}

\begin{lemma}[Completeness w.r.t.\ inductive characterization of regular semantics]\label{lem:reg-lp-complete}
For all $\cohyp$, $\extG{\Atm_1,\ldots,\Atm_n}{\EqSet}$ and $\sbt\in\gsol{\EqSet}$, \\
if $\validInd{\Loopcois{\lpProg}{\colpProg}}{\LoopJ{\appSubst{\cohyp}{\sbt}}{\appSubst{\Atm_i}{\sbt}}}$, for all $i\in 1..n$, then
$\colpNarrow{\lpProg}{\colpProg}{\cohyp}{\Atm_1,\ldots,\Atm_n}{\EqSet}{\EqSet'}$ and 
$\sbt{\sbtord}\asbt$, for some $\EqSet'$ \mbox{and $\asbt\in \gsol{\EqSet'}$}. 
\end{lemma}

\EZ{Soundness follows from \cref{lem:reg-lp-sound} and \cref{prop:loop-is}, as detailed below.
\begin{proofOf}{\cref{thm:reg-lp-sound}}
Let us assume $\colp{\lpProg}{\colpProg}{\emptyset}{\Goal}{\EqSet}{\EqSet'}$ with $\Goal = \Atm_1,\ldots,\Atm_n$, and consider $\sbt\in\gsol{\EqSet'}$. 
By \cref{lem:reg-lp-sound}, for all $i\in 1..n$, $\validInd{\Loopcois{\lpProg}{\colpProg}}{\LoopJ{\emptyset}{\appSubst{\Atm_i}{\sbt}}}$ holds, hence, by \cref{prop:loop-is}, we get $\appSubst{\Atm_i}{\sbt} \in \FlexReg{\lpProg}{\colpProg}$. 
Therefore, by \cref{def:lp-reg}, we get $\sbt\in\lpregans{\Goal}{\EqSet}$, as needed. 
\end{proofOf}}

\EZ{Analogously, completeness follows from \cref{lem:reg-lp-complete} and  \cref{prop:loop-is}, as detailed below.
\begin{proofOf}{\cref{thm:reg-lp-complete}}
Let $\Goal = \Atm_1,\ldots,\Atm_n$ and $\sbt\in\lpregans{\Goal}{\EqSet}$. 
Then, for all $i\in 1..n$, we have $\appSubst{\Atm_i}{\sbt}\in\FlexReg{\lpProg}{\colpProg}$ and, by \cref{prop:loop-is}, we get $\validInd{\Loopcois{\lpProg}{\colpProg}}{\LoopJ{\emptyset}{\appSubst{\Atm_i}{\sbt}}}$. 
Hence, the thesis follows by \cref{lem:reg-lp-complete}. 
\end{proofOf}}

\section{Related work and conclusion}\label{sect:related}
\EZ{We have provided a detailed formal account of an extension of logic programming where programs are enriched by coclauses, which can be used to tune the interpretation of predicates on non-well-founded structures. More in detail, following the same pattern as for standard logic programming, we have defined:
%formal account of an extension of coLP and coSLD resolution supporting flexible coinduction
\begin{itemize}
\item A declarative semantics (the union of all finite comodels which are subsets of a certain set of atoms determined by coclauses).
\item An operational semantics (a combination of standard SLD resolution and coSLD resolution) shown to be sound and complete with respect to the declarative semantics. 
\end{itemize}
As in the standard case, the latter provides a semi-algorithm. Indeed, concrete  strategies (such as breadth-first visit of the SLD tree) can be used to ensure that the operational derivation, if any, is found.  In this paper we do not deal with this part, however we expect it to be not too different from the standard case.}

\EZ{It has been shown \cite{AnconaD15} that, taking as declarative semantics the coinductive semantics (largest comodel), there is not even a semi-algorithm to check that an atom belongs to that semantics. Hence, there is no hope to find a complete operational semantics. On the other hand, our paper provides, for an extension of logic programming usable in pratice to handle non-well-founded structures, fully-developed foundations and results which are exactly the analogous of those for standard logic programming. }

%In this way, we can exprss fixed points of recursive definitions which
%are not necessarily the least, nor the greatest ones. 
%We have expressed  the declarative semantics as the largest comodel below a certain set of atoms determined by coclauses, 
%Then, we have proved that the latter is sound and complete w.r.t. the former restricted to finite comodels. 

CoLP has been initially proposed by Simon et al.\ \citeyear{Simon06,SimonMBG06,SimonBMG07} as a convenient sub-paradigm of logic programming to
model circularity; it was soon recognized the limitation of its expressive power that does not allow mutually recursive inductive and coinductive predicates, or predicates whose correct interpretation is neither the least, nor the greatest fixed point.

Moura et al.\ \citeyear{Moura13,Moura14} and Ancona \citeyear{Ancona13} have proposed implementations of coLP based on refinements of the Simon's original proposal
with the main aim of making them more portable and flexible.
Ancona has extended {coLP} by introducing a 
\textit{finally} clause, allowing the user to define the specific behavior of a predicate
when solved by coinductive hypothesis.
Moura's implementation is embedded in a tabled Prolog related to the implementation of Logtalk, and is based on a
mechanism similar to \textit{finally} clauses to specify customized behavior of predicates when
solved by coinductive hypothesis. While such mechanisms resemble coclauses, the corresponding formalization is
purely operational and lacks a declarative semantics and corresponding proof principles for proving correctness of predicate
definitions based on them.

Ancona and Dovier \citeyear{AnconaD15}
have proposed an operational semantics of coLP based on the big-step approach, which is simpler than the operational semantics
initially proposed by Simon et al.\, and proved it to be sound. They have also formally shown that there is no complete
procedure for deciding whether a regular goal belongs to the coinductive declarative semantics, but provided no completeness result restricted to regular derivations, neither mechanisms to extend coLP and make it more flexible.

Ancona et al.\ \citeyear{AnconaDZ17} were the first proposing a principled extension of coLP based on the notion
of cofact, with both a declarative and operational semantics; the latter is expressed in big-step style, following the approach of
Ancona and Dovier, and is proved to be sound w.r.t. the former. An implementation
is provided through a SWI-Prolog meta-interpreter. 

Our present work differs from the extension of coLP with cofacts mentioned above for the following novel contributions:
\begin{itemize}
\item we consider the more general notion of coclause, which includes the notion of cofact, but is a more expressive extension of coLP;
\item we introduce the notion of regular declarative semantics and prove coSLD resolution extended with
  coclauses is {sound and complete} w.r.t. the regular declarative semantics;
\item we show how {generalized} inference systems are closely related to logic programs with coclauses {and rely on this relationship to carry out proofs in a clean and principled way;}
\item we extend the implementation\footnote{See \url{https://github.com/davideancona/coLP-with-coclauses}{, where also examples of \refToSect{examples} are available.}} of the SWI-Prolog meta-interpreter to support coclauses.     
\end{itemize}

While coSLD resolution and its proposed extensions are limited by the fact that cycles must be detected  in derivations to allow resolution to succeed, a stream of work based on the notion
of \emph{structural resolution} \cite{KPS12-2,KJS17} (S-resolution for short)
aims to make coinductive resolution more powerful, by allowing to lazily detect infinite derivations
which do not have {cycles}.
In particular, recent results~\cite{Y17,KY17,BKL19} investigate how it is possible to integrate coLP cycle detection into S-resolution,
by proposing a comprehensive theory.
Trying to integrate S-resolution with coclauses is an interesting topic for future work aiming to make coLP even more flexible.

Another direction for further research consists in elaborating and extending the examples of logic programs with coclauses provided in \cref{sect:examples}, to formally prove their correctness, and experiment their effectiveness with the implemented meta-interpreter.

\bibliographystyle{acmtrans}
\bibliography{biblio}

\newpage
\appendix
\section{Proofs}

In this section we report proofs omitted in \cref{sect:correctness}.

\paragraph{Soundness} 
\EZ{We prove \cref{lem:reg-lp-sound}.}
To carry out the proof, we rely on \cref{prop:loop-is} and on 
the following proposition, stating that the inductive declarative semantics of a logic program coincides with the regular semantics of a logic program with no coclauses. 
The proof is given \EZ{by Dagnino \citeyear{Dagnino20}} in the general case of inference systems with corules.
\begin{proposition}\label{prop:is-reg-ind}
Let $\lpProg$ be a logic program, then 
$\Ind{\lpProg} = \FlexReg{\lpProg}{\emptyset}$. 
\end{proposition}

\begin{proofOf}{\cref{lem:reg-lp-sound}}
The proof is by induction on rules of \cref{fig:opsem}. 
\begin{description}
\item [\rn{empty}] 
There is nothing to prove. 

\item [\rn{step}] 
We have $\Goal = \Goal_1,\Atm_i,\Goal_2$, there is a fresh renaming $\clause{\bAtm}{\bAtm_1,\ldots,\bAtm_k}$ of a clause in $\lpProg$ such that $\Atm_i$ and $\bAtm$ are unifiable in $\EqSet$, that is, $\EqSet_1 = \EqSet\cup\UnEq{\Atm_i}{\bAtm}$ is solvable, and $\colpsem{\cohyp\cup\{\Atm_i\}}{\bAtm_1,\ldots,\bAtm_k}{\EqSet_1}{\EqSet_2}$ and $\colpsem{\cohyp}{\Goal_1,\Goal_2}{\EqSet_2}{\EqSet'}$ hold. 
Let $\sbt\in\gsol{\EqSet'}$, then, by induction hypothesis, we have, for all $j\in 1..n$ with $j\ne i$, $\validInd{\Loopcois{\lpProg}{\colpProg}}{\LoopJ{\appSubst{\cohyp}{\sbt}}{\appSubst{\Atm_j}{\sbt}}}$ holds. 
By \cref{prop:opsem-monotone}, we have $\EqSet_1\subseteq \EqSet_2\subseteq \EqSet'$, hence $\gsol{\EqSet'}\subseteq\gsol{\EqSet_2}\subseteq\gsol{\EqSet_1}$, thus $\sbt\in\gsol{\EqSet_2} \subseteq\gsol{\EqSet_1}$, and, since $\UnEq{\Atm_i}{\bAtm}\subseteq \EqSet_1$, $\sbt$ is a unifier of $\Atm_i$ and $\bAtm$, that is, $\appSubst{\Atm_i}{\sbt} = \appSubst{\bAtm}{\sbt}$. 
Then, by induction hypothesis, we also get, for all $j\in 1..k$, $\validInd{\Loopcois{\lpProg}{\colpProg}}{\LoopJ{\appSubst{(\cohyp\cup\{\Atm_i\})}{\sbt}}{\appSubst{\bAtm_j}{\sbt}}}$ holds. 
Since $\appSubst{(\cohyp\cup\{\Atm_i\})}{\sbt} = \appSubst{\cohyp}{\sbt}\cup\{\appSubst{\Atm_i}{\sbt}\}$ and $\clause{\appSubst{\bAtm}{\sbt}}{\appSubst{\bAtm_1}{\sbt},\ldots,\appSubst{\bAtm_k}{\sbt}} \in \Ground{\lpProg}$ and $\appSubst{\Atm_i}{\sbt} = \appSubst{\bAtm}{\sbt}$, by rule \rn{unfold} of \cref{def:loop-is}, we get that $\validInd{\Loopcois{\lpProg}{\colpProg}}{\LoopJ{\appSubst{\cohyp}{\sbt}}{\appSubst{\Atm_i}{\sbt}}}$ holds as well. 

\item [\rn{co-hyp}] 
We have $\Goal = \Goal_1,\Atm_i,\Goal_2$, there is an atom $\bAtm \in \cohyp$ that unifies with $\Atm_i$ in $\EqSet$, that is, $\EqSet_1 = \EqSet\cup\UnEq{\Atm_i}{\bAtm}$ is solvable, and $\lpsem{\Atm_i}{\EqSet_1}{\EqSet_2}$ and $\colpsem{\cohyp}{\Goal_1,\Goal_2}{\EqSet_2}{\EqSet'}$ hold. 
Let $\sbt\in \gsol{\EqSet'}$, then, by induction hypothesis, we get, for all $j\in 1..n$ with $j\ne i$, $\validInd{\Loopcois{\lpProg}{\colpProg}}{\LoopJ{\appSubst{\cohyp}{\sbt}}{\appSubst{\Atm_j}{\sbt}}}$ holds.
By \cref{prop:opsem-monotone}, we have $\EqSet_1\subseteq \EqSet_2\subseteq \EqSet'$, hence $\gsol{\EqSet'}\subseteq\gsol{\EqSet_2}\subseteq\gsol{\EqSet_1}$, thus $\sbt\in\gsol{\EqSet_2} \subseteq\gsol{\EqSet_1}$, and, since $\UnEq{\Atm_i}{\bAtm}\subseteq \EqSet_1$, $\sbt$ is a unifier of $\Atm_i$ and $\bAtm$, that is, $\appSubst{\Atm_i}{\sbt} = \appSubst{\bAtm}{\sbt}$. 
By induction hypothesis, we get $\validInd{\Loopcois{(\lpProg\cup\colpProg)}{\emptyset}}{\LoopJ{\emptyset}{\appSubst{\Atm_i}{\sbt}}}$, hence, by \cref{prop:loop-is} and \cref{prop:is-reg-ind}, we get $\appSubst{\Atm_i}{\sbt} \in \Ind{\lpProg\cup\colpProg}$. 
Furthermore, since $\appSubst{\Atm_i}{\sbt} = \appSubst{\bAtm}{\sbt}$ and $\bAtm \in \cohyp$, we have $\appSubst{\Atm_i}{\sbt} \in \appSubst{\cohyp}{\sbt}$. 
Therefore, by rule \rn{hp} of \cref{def:loop-is}, we get that $\validInd{\Loopcois{\lpProg}{\colpProg}}{\LoopJ{\appSubst{\cohyp}{\sbt}}{\appSubst{\Atm_i}{\sbt}}}$ holds as well. 
\end{description}
\end{proofOf}

%%We can prove \cref{thm:reg-lp-sound} relying on \cref{prop:loop-is}.
%\begin{proofOf}{\cref{thm:reg-lp-sound}}
%Let us assume $\colp{\lpProg}{\colpProg}{\emptyset}{\Goal}{\EqSet}{\EqSet'}$ with $\Goal = \Atm_1,\ldots,\Atm_n$ and consider $\sbt\in\gsol{\EqSet'}$. 
%By \cref{lem:reg-lp-sound} we have that, for all $i\in 1..n$, $\validInd{\Loopcois{\lpProg}{\colpProg}}{\LoopJ{\emptyset}{\appSubst{\Atm_i}{\sbt}}}$ holds, hence, by \cref{prop:loop-is}, we get $\appSubst{\Atm_i}{\sbt} \in \FlexReg{\lpProg}{\colpProg}$. 
%Therefore, by \cref{def:lp-reg}, we get $\sbt\in\lpregans{\Goal}{\EqSet}$, as needed. 
%\end{proofOf}

\paragraph{Completeness} \EZ{We need some preliminary results, then we prove \cref{lem:reg-lp-complete}.}

We start by observing a property of the operational semantics. 
In the following, we say that substitutions \DA{$\sbt_1$ and $\sbt_2$} are compatible, denoted by \DA{$\sbtcomp{\sbt_1}{\sbt_2}$} if, for all $\lpX\in\dom(\sbt_1)\cap\dom(\sbt_2)$, $\sbt_1(\lpX)=\sbt_2(\lpX)$, and 
we denote by $\sbt_1\sbtjoin\sbt_2$ the union of two substitutions, which is well-defined only for compatible substitutions. 
\FDComm{tagliato: if $\sbtcomp{\sbt_1}{\sbt_2}$ we denote by $\sbt_1\sbtjoin\sbt_2$ the union of the two substitutions defined by 
\[(\sbt_1\sbtjoin\sbt_2)(\lpX) = \begin{cases}
\sbt_1(\lpX) & \lpX\in\dom(\sbt_1) \\
\sbt_2(\lpX) & \text{otherwise}
\end{cases}\]
which is well-defined as the two substitutions are compatible. }
Note that, $\sbt_i\sbtord \sbt_1\sbtjoin\sbt_2$, for all $i=1,2$, by definition. 

\begin{proposition}\label{prop:ext-eq} 
Let $\extG{\Goal}{\EqSet_1}$ be a goal, $\sbt_1\in\gsol{\DA{\EqSet_1}}$, $\EqSet'_1$ such that $\EqSet_1\subseteq\EqSet'_1$ and  $\sbt_1\sbtord\asbt_1$, for some $\asbt_1\in\gsol{\EqSet'_1}$. 
If $\colp{\lpProg}{\colpProg}{\cohyp}{\Goal}{\EqSet_1}{\EqSet_2}$ and $\sbt_1\sbtord\sbt_2$, for some $\sbt_2\in\gsol{\EqSet_2}$, then 
there exists $\EqSet'_2$ such that $\EqSet_2\subseteq\EqSet'_2$, $\colpsem{\cohyp}{\Goal}{\EqSet'_1}{\EqSet'_2}$ and $\asbt_1\sbtord\asbt_2$, for some $\asbt_2\in\gsol{\EqSet'_2}$.
\end{proposition}
\begin{proof}
The proof is by induction on the big-step rules in \cref{fig:opsem}. 
\begin{description}
\item [\rn{empty}] 
We have $\EqSet_1=\EqSet_2$, hence the thesis follows \DA{by} taking $\EqSet'_2=\EqSet'_1$. 

\item [\rn{step}] 
We know that $\Goal = \Goal_1,\Atm,\Goal_2$, $\clause{\bAtm}{\bAtm_1,\ldots,\bAtm_n}$ is a fresh renaming of a clause in $\lpProg$, $\EqSet_1\cup\UnEq{\Atm}{\bAtm}$ is solvable, $\colpsem{\cohyp\cup\{\Atm\}}{\bAtm_1,\ldots,\bAtm_n}{\EqSet_1\cup\UnEq{\Atm}{\bAtm}}{\EqSet_3}$ and $\colpsem{\cohyp}{\Goal_1,\Goal_2}{\EqSet_3}{\EqSet_2}$ hold. 
We can assume that variables occurring in the selected clause do not belong to $\dom(\asbt_1)$ since such variables are fresh and $\dom(\asbt_1)$  is a finite set.
Since $\EqSet_1\cup\UnEq{\Atm}{\bAtm}\subseteq\EqSet_3 \subseteq \EqSet_2$ by \cref{prop:opsem-monotone}, we have $\sbt_2 \in\gsol{\EqSet_1\cup\UnEq{\Atm}{\bAtm}}$ and denote by $\sbt_1'$ the restriction of $\sbt_2$ to $\dom(\sbt_1)\cup\GVar{\bAtm}$. 
It is easy to see that, by construction, $\sbt_1\sbtord\sbt'_1$ and $\sbt'_1\in \gsol{\EqSet_1\cup\UnEq{\Atm}{\bAtm}}$ and $\sbtcomp{\sbt'_1}{\asbt_1}$, since $\dom(\sbt'_1)\cap\dom(\asbt_1) = \dom(\sbt_1)$ and $\sbt_1\sbtord\asbt_1$ by hypothesis. 
Hence, $\asbt'_1=\sbt'_1\sbtjoin\asbt_1$ is well-defined and $\sbt'_1\sbtord\asbt'_1$. 
Since $\EqSet_3\subseteq\EqSet_2$, $\sbt_2\in\gsol{\EqSet_3}$ and, if $\sbt'_2$ is the restriction of $\sbt_2$ to $\dom(\sbt'_1)\cup\GVar{\EqSet_3}$, then $\sbt'_2\in\gsol{\EqSet_3}$ as well, and $\sbt'_1\sbtord\sbt'_2$. 
Therefore, by induction hypothesis, we get that $\colpsem{\cohyp\cup\{\Atm\}}{\bAtm_1,\ldots,\bAtm_n}{\EqSet'_1\cup\UnEq{\Atm}{\bAtm}}{\EqSet'_3}$ holds and there is $\asbt'_2\in\gsol{\EqSet'_3}$ such that $\sbt'_2\sbtord\asbt'_2$, with $\EqSet_3\subseteq\EqSet'_3$. 
Since $\sbt'_2\sbtord\sbt_2$ and $\sbt'_2\sbtord\asbt'_2$, again by induction hypothesis, we get that $\colpsem{\cohyp}{\Goal_1,\Goal_2}{\EqSet'_3}{\EqSet'_2}$ holds and there is $\asbt_2\in\gsol{\EqSet'_2}$ such that $\sbt_2\sbtord\asbt_2$, with $\EqSet_2\subseteq\EqSet'_2$. 
Then, the thesis follows \DA{by} applying rule \rn{step}. 

\item [\rn{co-hyp}] 
We know that $\Goal = \Goal_1,\Atm,\Goal_2$, $\EqSet_1\cup\UnEq{\Atm}{\bAtm}$ is solvable for some $\bAtm\in\cohyp$, $\lpsem{\Atm}{\EqSet_1\cup\UnEq{\Atm}{\bAtm}}{\EqSet_3}$ and $\colpsem{\cohyp}{\Goal_1,\Goal_2}{\EqSet_3}{\EqSet_2}$ hold. 
Since $\EqSet_1\cup\UnEq{\Atm}{\bAtm}\subseteq\EqSet_3 \subseteq \EqSet_2$ by \cref{prop:opsem-monotone}, we have $\sbt_2 \in\gsol{\EqSet_1\cup\UnEq{\Atm}{\bAtm}}$ and denote by $\sbt_1'$ the restriction of $\sbt_2$ to $\dom(\sbt_1)\cup\GVar{\bAtm}$, but, as $\GVar{\bAtm}\subseteq \GVar{\EqSet_1}$, $\sbt_1\DA{\in}\gsol{\EqSet_1}$ and $\sbt_1\sbtord\sbt_2$, we get $\sbt'_1=\sbt_1$, thus $\sbt_1\in\gsol{\EqSet_1\cup\UnEq{\Atm}{\bAtm}}$. 
Hence, since $\sbt_1\sbtord\asbt_1$, we get $\asbt_1\in\gsol{\EqSet_1\cup\UnEq{\Atm}{\bAtm}}$ and so it also belongs to $\gsol{\EqSet'_1\cup\UnEq{\Atm}{\bAtm}}$, that is, in particular, $\EqSet'_1\cup\UnEq{\Atm}{\bAtm}$ is solvable. 
Since $\EqSet_3\subseteq\EqSet_2$, $\sbt_2\in\gsol{\EqSet_3}$ and, if $\sbt'_2$ is the restriction of $\sbt_2$ to $\dom(\sbt_1)\cup \GVar{\EqSet_3}$, then $\sbt'_2\in\gsol{\EqSet_3}$ as well, and $\sbt_1\sbtord\sbt'_2$. 
Therefore, by induction hypothesis, we get $\lpsem{\Atm}{\EqSet'_1\cup\UnEq{\Atm}{\bAtm}}{\DA{\EqSet'_3}}$ and there is $\asbt'_2\in\gsol{\EqSet'_3}$ such that $\sbt'_2\sbtord\asbt'_2$, with $\EqSet_3\subseteq\EqSet'_3$. 
Since $\sbt'_2\sbtord\sbt_2$ and $\sbt'_2\sbtord\asbt'_2$, again by induction hypothesis, we get that $\colpsem{\cohyp}{\Goal_1,\Goal_2}{\EqSet'_3}{\EqSet'_2}$ holds and there is $\asbt_2\in\gsol{\EqSet'_2}$ such that $\sbt_2\sbtord\asbt_2$, with $\EqSet_2\subseteq\EqSet'_2$. 
Then, the thesis follows \DA{by} applying rule \rn{co-hyp}. 

\end{description}
\end{proof}

\begin{lemma}\label{lem:ext-goal}
Let $\extG{\Goal}{\EqSet}$ be a goal, $\sbt\in\gsol{\EqSet}$ and $\EqSet_1$ and $\EqSet_2$ be sets of equations such that $\sbt\sbtord\sbt_1$ and $\sbt\sbtord\sbt_2$, for some $\sbt_1\in\gsol{\EqSet_1}$ and $\sbt_2\in\gsol{\EqSet_2}$. 
If $\colp{\lpProg}{\colpProg}{\cohyp}{\Atm}{\EqSet}{\EqSet_1}$ and $\colp{\lpProg}{\colpProg}{\cohyp}{\Goal_1,\Goal_2}{\EqSet}{\EqSet_2}$, then 
there exists $\EqSet_3$ such that $\colp{\lpProg}{\colpProg}{\cohyp}{\Goal_1,\Atm,\Goal_2}{\EqSet}{\EqSet_3}$ and $\sbt\sbtord\sbt_3$, for some $\sbt_3\in\gsol{\EqSet_3}$.
\end{lemma}
\begin{proof}
We sketch the proof. 
By \cref{prop:opsem-monotone}, we have $\EqSet\subseteq\EqSet_1$ and by \cref{prop:ext-eq} we get $\colp{\lpProg}{\colpProg}{\cohyp}{\Goal_1,\Goal_2}{\EqSet_1}{\EqSet_3}$, with $\EqSet_2\subseteq\EqSet_3$ and $\sbt_2\sbtord\sbt_3$, for some $\sbt_3\in\gsol{\EqSet_3}$.
By transitivity of $\sbtord$ we get $\sbt\sbtord\sbt_3$. 
Then, the thesis follows by case analysis on the last applied rule in the derivation of $\colp{\lpProg}{\colpProg}{\cohyp}{\Atm}{\EqSet}{\EqSet_1}$, \DA{by} replacing the premise $\colpsem{\cohyp}{\EList}{\EqSet_1}{\EqSet_1}$ with the judgement $\colp{\lpProg}{\colpProg}{\cohyp}{\Goal_1,\Goal_2}{\EqSet_1}{\EqSet_3}$. 
\end{proof}

\begin{lemma}\label{lem:bound-complete}
Let $\Atm$ be an atom and $\EqSet$ be a set of equations. 
For all $\sbt\in\gsol{\EqSet}$, 
if $\appSubst{\Atm}{\sbt}\in\Ind{\lpProg\cup\colpProg}$, then there exists $\EqSet'$ such that 
$\colp{\lpProg\cup\colpProg}{\emptyset}{\emptyset}{\Atm}{\EqSet}{\EqSet'}$ and 
$\sbt\sbtord\asbt$, for some $\asbt\in \gsol{\EqSet'}$. 
\end{lemma}
\begin{proof}
The proof is by induction on the derivation of $\appSubst{\Atm}{\sbt}$ in $\Ground{\lpProg\cup\colpProg}$. 
Let $\clause{\Atm'}{\Atm_1,\ldots,\DA{\Atm_n}} \in \Ground{\lpProg\cup\colpProg}$ be the last applied rule in the finite derivation of $\appSubst{\Atm}{\sbt}$, hence we have $\Atm' = \appSubst{\Atm}{\sbt}$. 
By definition of $\Ground{\lpProg\cup\colpProg}$, we know there is \DA{a} fresh renaming of a clause in $\lpProg\cup\colpProg$, denote it by $\clause{\bAtm}{\bAtm_1,\ldots,\bAtm_n}$, and a substitution $\sbt'$ such that $\appSubst{\bAtm}{\sbt'} = \Atm'$ and $\appSubst{\bAtm_i}{\sbt'} = \Atm_i$, for all $i \in 1..n$. 
Since the variables in this clause are fresh, we can assume $\dom(\sbt)\cap\dom(\sbt') = \emptyset$, hence $\sbt'' = \sbt\sbtjoin\sbt'$ is well-defined and, by construction, we have 
$\sbt\sbtord\sbt''$, thus $\sbt''\in\gsol{\EqSet}$, and 
$\appSubst{\Atm}{\sbt''} = \appSubst{\bAtm}{\sbt''}$, that is, $\sbt'' \in \gsol{\UnEq{\Atm}{\bAtm}}$. 
As a consequence $\sbt''$ is a solution of $\EqSet\cup\UnEq{\Atm}{\bAtm}$, hence, 
by induction hypothesis, we get that, for all $i\in 1..n$, $\lpsem{\bAtm_i}{\EqSet\cup\UnEq{\Atm}{\bAtm}}{\EqSet_i}$ holds and $\sbt''\sbtord\asbt_i$, for some $\asbt_i\in\gsol{\EqSet_i}$. 

By applying $n$ times \cref{lem:ext-goal}, we get $\lpsem{\bAtm_1,\ldots,\bAtm_n}{\EqSet\cup\UnEq{\Atm}{\bAtm}}{\EqSet'}$ and $\sbt''\sbtord\asbt$, for some $\asbt\in\gsol{\EqSet'}$. 
Then, the thesis follows by applying rule \rn{step} to this judgement and $\lpsem{\EList}{\EqSet'}{\EqSet'}$. 
\end{proof}

\begin{lemma} \label{lem:atom-complete}
For all $\cohyp$ and $\extG{\Atm}{\EqSet}$,\\
%Let $\EqSet$ be a set of equations.
if $\validInd{\Loopcois{\lpProg}{\colpProg}}{\LoopJ{\appSubst{\cohyp}{\sbt}}{\appSubst{\Atm}{\sbt}}}$,  and $\sbt\in\gsol{\EqSet}$,  then
$\colpNarrow{\lpProg}{\colpProg}{\cohyp}{\Atm}{\EqSet}{\EqSet'}$ and 
$\sbt{\sbtord}\asbt$, for some $\EqSet'$ \mbox{and $\asbt\in \gsol{\EqSet'}$}. 
\end{lemma}
\begin{proof}
The proof is by induction on the derivation of $\LoopJ{\appSubst{\cohyp}{\sbt}}{\appSubst{\Atm}{\sbt}}$ (see \cref{def:loop-is}). 
\begin{description}
\item [\rn{hp}]
We know that $\appSubst{\Atm}{\sbt}\in\appSubst{\cohyp}{\sbt}$ and $\appSubst{\Atm}{\sbt}\in\Ind{\lpProg\cup\colpProg}$, that is, $\sbt\in\lpbdans{\Atm}{\EqSet}$. 
By \cref{lem:bound-complete}, we get $\lpsem{\Atm}{\EqSet}{\EqSet_1}$ and $\sbt\sbtord\sbt_1$, for some $\sbt_1\in \gsol{\EqSet_1}$. 
Since $\appSubst{\Atm}{\sbt}\in\appSubst{\cohyp}{\sbt}$, we know that there is $\bAtm\in\cohyp$ such that $\appSubst{\DA{\Atm}}{\sbt} = \appSubst{\bAtm}{\sbt}$, that is $\sbt\in \gsol{\UnEq{\Atm}{\bAtm}}$, thus $\sbt \in \gsol{\EqSet\cup\UnEq{\Atm}{\bAtm}}$. 
Therefore, by \cref{prop:ext-eq}, we get $\lpsem{\Atm}{\EqSet\cup\UnEq{\Atm}{\bAtm}}{\EqSet'}$ and $\sbt\sbtord\sbt_1\sbtord\asbt$, for some $\asbt\in\gsol{\EqSet'}$. 
The thesis follows by applying rule \rn{co-hyp} to this judgement and $\colpsem{\cohyp}{\EList}{\EqSet'}{\EqSet'}$. 
 
\item [\rn{rule}] 
We know there is a $\clause{\Atm'}{\Atm_1,\ldots,\Atm_n} \in \Ground{\lpProg}$ such that $\Atm' = \appSubst{\Atm}{\sbt}$ and $\validInd{\Loopcois{\lpProg}{\colpProg}}{\LoopJ{\appSubst{\cohyp}{\sbt}\cup\{\appSubst{\Atm}{\sbt}\}}{\Atm_i}}$ is derivable, for all $i\in 1..n$. 
By definition of $\Ground{\lpProg\cup\colpProg}$, we know there is \DA{a} fresh renaming of a clause in $\lpProg\cup\colpProg$, denote it by $\clause{\bAtm}{\bAtm_1,\ldots,\bAtm_n}$, and a substitution $\sbt'$ such that $\appSubst{\bAtm}{\sbt'} = \Atm'$ and $\appSubst{\bAtm_i}{\sbt'} = \Atm_i$, for all $i \in 1..n$. 
Since the variables in this clause are fresh, we can assume $\dom(\sbt)\cap\dom(\sbt') = \emptyset$, hence $\sbt'' = \sbt\sbtjoin\sbt'$ is well-defined and, by construction, we have 
$\sbt\sbtord\sbt''$, thus $\sbt''\in\gsol{\EqSet}$, and 
$\appSubst{\Atm}{\sbt''} = \appSubst{\bAtm}{\sbt''}$, that is, $\sbt'' \in \gsol{\UnEq{\Atm}{\bAtm}}$ and 
$\appSubst{\cohyp}{\sbt}\cup\{\appSubst{\Atm}{\sbt}\} = \appSubst{(\cohyp\cup\{\Atm\})}{\sbt''}$. 
As a consequence $\sbt''$ is a solution of $\EqSet\cup\UnEq{\Atm}{\bAtm}$, hence, 
by induction hypothesis, we get that, for all $i\in 1..n$, $\colpsem{\cohyp\cup\{\Atm\}}{\bAtm_i}{\EqSet\cup\UnEq{\Atm}{\bAtm}}{\EqSet_i}$ holds and $\sbt''\sbtord\asbt_i$, for some $\asbt_i\in\gsol{\EqSet_i}$. 

By applying $n$ times \cref{lem:ext-goal}, we get $\colpsem{\cohyp\cup\{\Atm\}}{\bAtm_1,\ldots,\bAtm_n}{\EqSet\cup\UnEq{\Atm}{\bAtm}}{\EqSet'}$ and $\sbt''\sbtord\asbt$, for some $\asbt\in\gsol{\EqSet'}$. 
By transitivity we get $\sbt\sbtord\asbt$.
Then, the thesis follows by applying rule \rn{step} to this judgement and $\lpsem{\EList}{\EqSet'}{\EqSet'}$. 

\end{description}
\end{proof}

\begin{proofOf}{\cref{lem:reg-lp-complete}} 
The proof is by induction on the number of atoms in $\Goal = \Atm_1,\ldots,\Atm_n$. 
If $\Goal = \EList$, then the thesis follows by rule \rn{empty}, taking $\EqSet' = \EqSet$ and $\asbt = \sbt$. 
If $\Goal = \Goal_1,\Atm,\Goal_2$, 
we know by hypothesis that 
$\validInd{\Loopcois{\lpProg}{\colpProg}}{\LoopJ{\appSubst{\cohyp}{\sbt}}{\appSubst{\Atm}{\sbt}}}$,
hence, by \cref{lem:atom-complete} we get $\colp{\lpProg}{\colpProg}{\cohyp}{\Atm}{\EqSet}{\EqSet_1}$ and $\sbt\sbtord\sbt_1$, for some $\sbt_1\in\gsol{\EqSet_1}$.
By induction hypothesis, we get $\colpsem{\cohyp}{\Goal_1,\Goal_2}{\EqSet}{\EqSet_2}$ and $\sbt\sbtord\sbt_2$, for some $\sbt_2\in \gsol{\EqSet_2}$, 
therefore, by \cref{lem:ext-goal}, we get $\colpsem{\cohyp}{\Goal_1,\Atm,\Goal_2}{\EqSet}{\EqSet'}$ and $\sbt\sbtord\asbt$ for some $\asbt \in \gsol{\EqSet'}$.  
\end{proofOf}

\label{lastpage}
\end{document}